\documentclass[authoryear,12pt,3p,1in]{jowarticle}

\usepackage[english]{babel}
\usepackage[utf8]{inputenc}
\usepackage{amsmath,amsfonts,amssymb,amsthm}
\usepackage{graphicx}
\usepackage[colorinlistoftodos]{todonotes}
\usepackage{natbib}
\usepackage{blindtext}
\usepackage{subcaption}
\usepackage{caption}
\usepackage{setspace}
\usepackage{tikz-cd}
\usepackage{tensor}
\usepackage{accents}
\usepackage[shortlabels]{enumitem}
\newlength{\dhatheight}

\newtheorem{thm}{Theorem}

\makeatletter
\renewcommand\section{\@startsection{section}{1}{\z@}{-3.25ex plus -1ex minus -.2ex}{1.5ex plus .2ex}{\normalsize\bf}}
\renewcommand\subsection{\@startsection{subsection}{2}{\z@}{-3.25ex plus -1ex minus -.2ex}{1.5ex plus .2ex}{\normalsize\bf}}
\renewcommand\subsubsection{\@startsection{subsubsection}{3}{\z@}{-3.25ex plus -1ex minus -.2ex}{1.5ex plus .2ex}{\normalsize\bf}}
\makeatother

\begin{document}
\begin{frontmatter}
\title{Natural Theories}
\author{James Owen Weatherall}\ead{james.owen.weatherall@uci.edu} 
\address{Department of Logic and Philosophy of Science \\ University of California, Irvine}
\author{Eleanor March}\ead{eleanor.march@philosophy.ox.ac.uk}
\address{Faculty of Philosophy \\ University of Oxford}

\date{\today}

\begin{abstract}
We consider the class of physical theories whose dynamics are given by \emph{natural equations}, which are partial differential equations determined by a functor from the category of $n$-manifolds, for some $n$, to the category of fiber bundles, satisfying certain further conditions.  We show how the theory of natural equations clarifies several important foundational issues, including the status and meaning of minimal coupling, symmetries of theories, and background structure. We also state and prove a fundamental result about the initial value problem for natural equations.
\end{abstract}
\end{frontmatter}

\onehalfspacing
\section{Introduction} \label{sec:intro}

\citet{March+Weatherall} have recently defended a view on which the oft-invoked desideratum that a physical theory be `generally covariant' should be seen as requiring (at least) that the mathematical objects used to represent physical quantities and their states be functorial over smooth manifolds.\footnote{There is a straightforward sense in which this proposal is a development of arguments from \citet[p.~48]{Misner+etal}, who define general covariance as the requirement that `every physical quantity must be describable by a geometric object' \'a la \citet{Nijenhuis}.} They make this idea precise using the formalism of \emph{natural bundles} \citep{Nijenhuis,Salvioli,Terng,Kolar+etal}, which, as we discuss below, are functors from certain categories of smooth manifolds to categories of fiber bundles. The justification for this proposal is that general covariance, however it is explicated, apparently requires that salient mathematical objects have a well-defined behavior under the action of diffeomorphisms. In other words, before one can assess whether some object has the `correct' behavior under diffeomorphisms (or coordinate transformations), one must first say how diffeomorphisms act at all. Functoriality is an attractive and mathematically well-motivated way to implement this requirement.

As March and Weatherall emphasize, however, requiring that the objects used to represent physical quantities be valued in natural bundles is at most a necessary condition for general covariance. As they put it, `any full account' of general covariance `would require that the principal claims and relationships asserted by a theory are preserved under the actions of diffeomorphisms on the objects concerned' (p.~4). Our goal here is to complete that project.\footnote{Similar ideas are discussed by \citet{Fatibene+Francaviglia}, though they focus on variational theories and do not discuss the main applications of the idea that we introduce here.  See also \citet{FletcherWeatherallpart2}.} The idea is to consider (finite-order) partial differential equations on natural bundles that are themselves `natural', in the sense of being functorial over smooth manifolds. A natural equation on a natural bundle would be one that depends only on the natural bundle structure and derivatives of that structure to finite order, in such a way that if the equation is satisfied for some section of some bundle, then it is also satisfied by the pushforward of that section along any smooth embedding. This proposal requires one to associate a given equation with all (suitable) smooth manifolds of a given dimensions and to specify how the equation behaves under the action of smooth embeddings. A \emph{natural theory} would be one whose fields and equations are natural.

Natural theories are of interest for a number of reasons, beyond being a plausible explication of general covariance. There are important examples of natural theories, such as general relativity, special relativity, or, in five or more dimensions, Gauss-Bonnet gravity. Not all physical theories are natural, but as we discuss below, a non-natural theory can often be reconceived as natural by moving to a different (natural) bundle. This construction, which we call `naturalization', can be though of as a kind of Kretschmannization procedure \citep{Kretschmann}, where one identifies background structure and makes explicit how that structure itself transforms under the action of smooth embeddings. The famously vexed criterion of `minimal coupling' for a system of equations can also be explicated in this setting, as a claim about the minimal background structure needed to naturalize a system of equations in the context of relativity theory.\footnote{Thus, this approach can be seen to address concerns raised by \citet{Fletcher2020}, \citet{WeatherallDogmas}, and \citet{FletcherWeatherallpart2} about the hyperintensionality of minimal coupling, and it suggests a path forward for capturing criteria on classical field theories that are usually expressed syntactically.}  Naturalization also provides a new perspective on the role and status of symmetries of physical theories.  We will disambiguate three senses of symmetry and use them to clarify the status of Earman's famous matching principles \citep{WEST}.

We will also prove two important theorems about natural theories. The first, which is somewhat technical, is motivated in section \ref{sec:natTheories} and stated and proved in \ref{app:}. The second is easy to state, and of immediately relevance to the philosophical literature. Let us say that a system of differential equations admits a well-posed initial value problem if, given some suitable initial data, there exists a unique solution, at least locally, that coincides with that initial data. We then have the following:
\begin{thm}
No sufficiently rich natural system of equations admits a well-posed initial value problem.\end{thm}
\noindent (We will define `sufficiently rich' below.) This result is of intrinsic interest, since it indicates that there is an incompatibility between naturality and a certain conception of determinism.  But perhaps more interesting, for a philosophical audience, is that the proof of this theorem is a version of the infamous hole argument \citep{Earman+Norton,Norton,Pooley}. We take this result to provide insight into the structure of natural equations, but also to advance the literature on the status of the hole argument. In particular, it isolates the sense in which the hole argument rests on an `error' \citep[c.f.][]{WeatherallHoleArg,Bradley+Weatherall}, from the point of view of ordinary mathematical practice, which is that the sense of uniqueness of solution at issue in the argument not natural. One also sees a sense in which the argument is perfectly legitimate, insofar as it arises in the proof of a theorem -- albeit one that shows how a certain construction fails to be natural.\footnote{There is a close connection here to recent work on determinism by \citet{Halvorson+etal}, though we will postpone a detailed discussion of that connection to future work.}  

The paper will proceed as follows. We begin in the next section with some background preliminaries on natural bundles and jet bundles. In section \ref{sec:natTheories}, we will introduce natural equations and natural theories. Section \ref{sec:naturalization} will introduce naturalization, and section \ref{sec:MinCoup} will apply it to minimal coupling.  Section \ref{sec:Earman} discusses symmetries of theories.  Then in section \ref{sec:HoleArg} we will prove Theorem 1 and discuss its significance. Section \ref{sec:conclusion} concludes.

\section{Natural Bundles and Jet Bundles}\label{sec:preliminaries}

Natural bundles are a way of capturing the idea of a geometric object as developed by \citet{Nijenhuis} and others.\footnote{The standard and encyclopedic resource on natural bundles and natural operators is \citep{Kolar+etal}.} Very roughly, natural bundles are bundles that depend on only the structure of their base space, in the sense that given any appropriate smooth manifold, the bundle can always be constructed; and any appropriate smooth map on the base space `lifts' uniquely to a bundle morphism.  To make this idea precise, let $\mathcal{M}_n$ denote the category of smooth, $n$-dimensional manifolds, with smooth embeddings as morphisms.\footnote{Our definition of $\mathcal{M}_n$ here is similar to \citet{Palais+Terng}. An alternative approach, following \citet{Kolar+etal}, would be to define the category of $n$-manifolds to have local diffeomorphisms as arrows.} Let $\mathcal{FB}$ be the category whose objects are smooth fiber bundles and whose morphisms are smooth bundle morphisms. A \emph{natural bundle} (over $n$-manifolds) is a functor $F:\mathcal{M}_n\rightarrow \mathcal{FB}$ such that (1) for every object $M$ of $\mathcal{M}_n$, $FM$ is a bundle whose base space is $M$; and (2) for every morphism $\varphi:M\rightarrow N$ of $\mathcal{M}_n$, $F\varphi$ is of the form $(\varphi_*,\varphi)$, where the maps $\varphi_*$ induces from fibers of $FM$ to fibers of $FN$ are diffeomorphisms.\footnote{The notation is suggestive, but not all natural bundle functors involve the pushforward construction.}

Functoriality, here, is what enforces the two informal conditions above.  Natural bundles associate smooth bundles with (any) smooth $n$-manifolds and smooth bundle morphisms with smooth embeddings (whilst preserving composition and identity). Meanwhile, the restriction to `appropriate' smooth manifolds and smooth maps gets realized in the definition of the category $\mathcal{M}_n$.  Note that we have opted for a general category in our definition of natural bundle, but for some purposes we will restrict to a more specialized (full) subcategory.  In particular, we will often consider the category $\bar{\mathcal{M}}_n$ of smooth $n-$manifolds admitting a Lorentzian metric \citep{Geroch+Horowitz,ONeill}.

To get an intuitive handle on this idea, consider, as an example, tangent bundles. Any smooth manifold $M$ determines a bundle $TM\rightarrow M$, and any well-behaved smooth map $\varphi:M\rightarrow N$ on the base space determines a smooth map from $TM$ to $TN$ via the pushforward construction. Thus tangent bundles realize the two properties we have identified. Other examples include cotangent bundles, (tangent) frame bundles, tensor bundles, bundles of $n$-forms for fixed $n$, and so on.  Note that officially, a natural bundle is not a bundle, but a functor from the category of smooth $n$-manifolds to fiber bundles. Nevertheless, we will (by abuse of language) continue to use the term `natural bundle' to refer to \emph{both} the functor and the image of objects $M$ in $\mathcal{M}_n$ under that functor. (This is similar to how `tangent bundle' is used to refer both to the general construction and individual instances.)

Now we move on to discuss jet bundles. Recall that if $M$ and $N$ are smooth manifolds, and $f:U\rightarrow N$, $g:V\rightarrow N$ are smooth maps defined on open neighbourhoods $U$, $V$ of some $p\in M$, then $f$ and $g$ are said to be $k$-equivalent at $p$ iff they agree on all their partial derivatives up to order $r$ at $p$ (in any, and thus every, local coordinate charts containing $p$, $f(p)$, $g(p)$). A $k$-jet at $p$ is an equivalence class $[f]_p$ of smooth maps that are $k$-equivalent at $p$, and the $k$-jet at $p$ containing $f$ is denoted $j^k_pf$. Now let $B\overset{\pi}{\rightarrow}M$ be a smooth fiber bundle. Then we can construct a smooth bundle $J^kB$ (the $k$th jet bundle of $B$) whose total space is the space of all $k$-jets of local sections of $B$, for all points in $M$. This bundle is equipped with a variety of projection maps: for any $\ell\leq k$ one has the smooth map $\pi^k_{ell}:J^kB\rightarrow J^{\ell}B$, defined by $\pi^k_{\ell}(j^k_p\psi)=j^{\ell}_p\psi$, which (together with $\pi$) make it a smooth bundle over $M$ and over $J^{\ell}B$. 

It is important for what follows to observe that if $X$ is a natural bundle, then $J^kX$ is also natural. To see this, consider that if $X$ is natural, then any smooth embedding $\varphi: M\rightarrow N$ will determine a smooth bundle morphism $(\varphi^*,\varphi):XM\rightarrow XN$. But any smooth bundle morphism $(\varphi^*,\varphi)$ on $XM$ lifts uniquely to a smooth bundle morphism $j^k\varphi$ on $J^kXM$ (its $k$th jet prolongation) via $j^k\varphi(j^k_p\psi)=j^k_{\varphi(p)}\varphi^*(\psi(\varphi^{-1}(p)))$.  This map, restricted to fibers, is a diffeomorphism.  By the same reasoning, one can define a functor $J^1:\mathcal{FB}\rightarrow\mathcal{FB}$ taking every bundle to its first prolongation, which is such that for any natural bundle $X:\mathcal{M}_n\rightarrow \mathcal{FB}$, $J^1X$ is also a natural bundle.\footnote{Here $J^1X=J^1\circ X$.  We suppress the composition symbol `$\circ$' in what follows.} The first jet bundle of $XM\rightarrow M$ is precisely the image of $M$ under $J^1 X$. Similarly, for any $k\geq 1$, we can define $k$th prolongation functors $J^k:\mathcal{FB}\rightarrow\mathcal{FB}$, and then we find that for any natural bundle $X$, $J^kX$ is natural. For any $k,\ell$ with $k>\ell$, there exists a canonical natural transformation $\pi^k_{\ell}:J^k\rightarrow J^{\ell}$ whose factors are the maps $\pi^k_{\ell}$ projecting $J^k$ down to $J^{\ell}$.
 
Note that, despite the suggestive notation, we do \emph{not} have $J^k = \underbrace{J^1\cdot J^1}_{k times}$. Instead, what we find is that for any $k$, there is a canonical natural transformation $\eta^k:J^k\rightarrow \underbrace{J^1\cdot J^1}_{k times}$ taking $k$ jets into 1st jets of 1st jets, etc.~$k$ times.\footnote{Namely, we have $\eta^k(j^k_p\psi)=j^1_p(\underbrace{j^1(j^1...(j^1}_{k-1 times}\psi)))$.} The extends, for any natural bundle $X$, to a natural transformation from $J^kX \rightarrow \underbrace{J^1\cdot J^1}_{k times} X$.  To see why this would be, just consider the special case of scalar fields, and note that at a point, any covector can arise as the derivative of a scalar field. But not every derivative of a covector field at a point can arise as the second derivative of a scalar field. For that, a further constraint is needed: the second derivative must be symmetric. 

\section{Natural Equations and Natural Theories}\label{sec:natTheories}

We now turn to differential equations on natural bundles.  Recall, first, that a $k$th order system of partial differential equations on sections of a bundle $B\rightarrow M$ can be viewed as a closed submanifold $E$ of the (total space of the) $k$th jet bundle over $B$. To see this, note that a partial differential equation is, in particular, an equation, which relates field values and their derivatives to kth order; and a jet bundle has, as fibers over $M$, field values and their first k derivatives.  But any equation (or system of equations) on coordinates of a manifold determines a closed submanifold of points whose coordinate values satisfy the equation.  A solution to a differential equation is a (local) section $\varphi:U\rightarrow B$ of the bundle whose $k$th prolongation lies everywhere within $E$.

It is tempting to go the other direction and define a $k$th order partial differential equations as a closed embedded submanifold of the total space of a $k$th order jet bundle. But this would be too general. Instead, we impose some further (minimal) regularity conditions on the systems of equations we consider here. Fix a bundle $B\xrightarrow{\pi}M$.  In what follows, a $k$th order system of partial differential equations will be a closed submanifold $E$ of $j^kB$ satisfying the following conditions:
\begin{enumerate}
    \item $j^k\pi_{|E}:E\rightarrow M$ is surjective.
    \item $\mathrm{d}\pi_{|E}:TE\rightarrow TM$ is (pointwise) surjective.
\end{enumerate}
Or, putting (1) and (2) together, we require that $j^k\pi_{|E}:E\rightarrow M$ be a fibered manifold. Condition (1) is necessary for $E$ to have solutions about every point in $M$. Condition (2) is necessary for $E$ to have solutions through every point in $E$. 

These general remarks apply to natural bundles as well.  That is, given any natural bundle $X:\mathcal{M}_n\rightarrow\mathcal{FM}$ and $n$-manifold $M$, we can define a system of PDEs on $XM\rightarrow M$ as a submanifold of a jet bundle over $XM$ of appropriate order, satisfying the regularity conditions above.  But when working with natural bundles, we have the resources to consider equations in a somewhat different way.  A natural bundle assigns a certain bundle structure to \textit{every} $n$-manifold in $\mathcal{M}_n$; an equation, on the other hand, is a submanifold of a jet bundle over just one of those assignments. This observation suggests the possibility of defining equations on all $n$-manifolds in $\mathcal{M}_n$, in a uniform and functorial way.  Doing so would make use of the full natural bundle structure, and it would provide a precise characterization of when one has `the same' equation on different $n$-manifolds.

To motivate this idea from a physics perspective, consider that in general relativity, and relativistic field theory more generally, we consider spacetimes based on various manifolds, and yet we are able to make sense of the same equations on those different manifolds.  This suggests that the right setting, or level of generality, for, say, Maxwell's equations and Einstein's equation, is not the jet bundle over a particular $n$-manifold, but rather the family, or better still, category, of jet bundles over arbitrary $n$-manifolds.  That is, the partial differential equations of physics are assignments of submanifolds to jet bundles that range over the full space of possible spacetime manifolds, and not just one manifold.  

How should we think about the `same equations' across different manifolds?  One flat-footed response would be to note that in ordinary practice, the equations of general relativity and relativistic physics are all local, in the sense that they depend just on tensor fields and their derivatives in neighborhoods of a point.  And of course, at that local level, one can think of any $n$-manifold as diffeomorphic to an open region of $\mathbb{R}^n$, where we typically have canonical definitions of the equations of interest.\footnote{We are glossing over several important points discussed by \citet{FletcherWeatherallpart1}, concerning how to treat equations that depend on a background structure, such as a metric.  We will return to that issue in section \ref{sec:MinCoup}.}   From this perspective, it might seem that there is no mystery as to how we can think of the same equation on different manifolds, at least locally. But this observation obscures two important issues.  One of them is that to import an equation from $\mathbb{R}^n$ in the way implied by this construction, one must be able to associate a submanifold of a jet bundle over $\mathbb{R}^n$ with a corresponding submanifold of the jet bundle over regions of other manifolds.  This is always possible, but the \emph{reason} it is possible is that jet bundles are natural, and so any smooth map from a region of $\mathbb{R}^n$ into a region of another manifold $M$ lifts to a map between jet bundles.  The second issue is that while we can and often do think of equations as local in this way, there is a further question about when the local submanifolds these equations define can be patched together to form a consistent submanifold of the jet bundle over the entire manifold.  

These considerations motivate the following definition.  A \emph{natural (k-th order) equation} $(E,e)$ on a natural bundle $X:\mathcal{M}_n\rightarrow \mathcal{FB}$ is a functor $E:\mathcal{M}_n\rightarrow\mathcal{FB}$ and a monic natural transformation $e:E\rightarrow J^k X$.\footnote{Note that the components of any natural transformation between natural bundles are smooth bundle morphism whose action on the base space is the identity.} 
In other words, a natural equation on a natural bundle is an assignment, to every manifold, of a submanifold of some (finite order) jet bundle over the natural bundle over that manifold, in such a way that it is preserved by smooth mappings between $n$-manifolds and their lifts via $X$ and $J^\ell X$, for every $\ell\leq k$.  Figure \ref{fig:naturalEqs} depicts how these natural transformations behave.

\begin{figure}\centering
\begin{tikzcd} EM \arrow[rr, "E\varphi"] \arrow[dd, hook, "e_M"] & & EN \arrow[dd, hook, "e_N"] \\ & & \\
J^kXM \arrow[rr, "J^kX\varphi"] \arrow[d, "j^k\pi_M"] & & J^kXN \arrow[d, "j^k\pi_N"] \\
\vdots \arrow[rr] \arrow[d, "j^2\pi_M"] & & \vdots \arrow[d, "j^2\pi_N"]\\
J^1XM \arrow[rr, "J^1X\varphi"] \arrow[dd, "j^1\pi_M"] & & J^1XN \arrow[dd, "j^1\pi_N"]\\  & & \\
XM \arrow[rr, "X\varphi"] \arrow[dd, "\pi_M"] & & XN \arrow[dd, "\pi_N"]\\  & & \\
            M \arrow[rr, "\varphi"] & & N\end{tikzcd} 
            \caption{\label{fig:naturalEqs} The naturality diagrams for natural equations over natural bundles.}
            \end{figure}

There is a feature of our definition that deserves further comment.  Our general definition of partial differential equations required the submanifold $E$ to be a fibered manifold over $M$. However, we have defined natural equations as functors from manifolds to \emph{fiber bundles}.  This implicitly imposes a further regularity condition on the submanifolds of the jet bundle that we consider as equations:
\begin{itemize}
    \item[3.] $j^k\pi_{|E}:E\rightarrow M$ is a fiber bundle. 
\end{itemize}
One could proceed differently, and consider natural equations as functors, not to the category of fiber bundles, but to the category of fibered manifolds.  But in fact, nothing depends on the choice.  This is because, were we to define natural equations in this more general way, condition (3) would hold automatically as a consequence of naturality.  (Indeed, conditions (1) and (2) need not even be imposed.  See \ref{app:} for a precise statement of these claims and a proof.) So in the special case of natural equations, we can talk about equations as subbundles of natural bundles without loss of generality.

Finally, we define a \emph{natural theory} (or, a \emph{generally covariant field theory}) $(X,E,e)$ to consist of a natural bundle $X$ and a natural equation $(E,e)$ on that bundle.  This definition captures the idea that any (generally covariant) physical theory specifies a space of possible field values, and a system of differential equations on those fields, that are compatible with the manifold structures of spacetime, in the sense that the fields have a well-defined action under smooth maps, and the equations imposed on those fields are invariant under that action.  We can see this definition as fulfilling the promisory note of \citet{March+Weatherall}, who argue that a generally covariant theory should be set on a natural bundle, but also suggest that to complete their characterization of general covariance, one would need to say what it means for the equations of the theory to be preserved under the lifts of smooth maps to those natural bundles.  Our definition of natural theories does precisely that.

\section{Background Structure and Naturalization}\label{sec:naturalization}

Given the above definitions of natural equations and natural theories, one might reasonably ask: are standard examples of field equations (and field theories) considered in physics natural? In general, the answer to this question is `no'. Consider, for example, the source-free Maxwell equations. First, observe that the bundle of two-forms over arbitrary smooth 4-manifolds is a natural bundle $Z: \mathcal{M}_n\rightarrow\mathcal{FB}$. Thus, given a relativistic spacetime, $(M, g_{ab})$, the source-free Maxwell equations determined by $g_{ab}$ and its associated Levi-Civita connection $\nabla$ is an (affine) subbundle $E$ of $J^1ZM\rightarrow ZM\rightarrow M$. Now let $\varphi:M\rightarrow M$ be a diffeomorphism, and consider the action on $E$ of the lift of $\varphi$ under $J^1ZM$. If Maxwell's equations were natural, it would follow that $E$ maps to itself under the lift. It is immediate to see that this condition holds iff $\varphi^*g_{ab}=g_{ab}$. In other words, the Maxwell system on $ZM$ is not natural, since $E$ is not preserved by those smooth mappings that are not isometries of $g_{ab}$.

There is a simple intuition behind this. Maxwell's equations do not depend only on the choice of Faraday tensor, they also depend on the metric.\footnote{They also depend on the Levi-Civita derivative operator.  However, one way of capturing the fundamental theorem of (pseudo-)Riemannian geometry is that, for $n=2$ or $n> 3$, the Levi-Civita derivative operator is the unique first-order natural operator from bundles of (pseudo)-Riemannian metrics to arbitrary natural bundles \citep[\S\S 28.15 \& 33.19]{Kolar+etal}, which provides a precise sense in which the metric itself determines the Levi-Civita operation as long as all operations are natural. (The cases for $n=1$ and $n=3$ are interesting, but not directly relevant.)} Thus, the system is not preserved by those diffeomorphisms which do not fix the metric. But this also gives us an idea of how to `naturalize' Maxwell's equations, namely by making the dependence on the metric explicit. That is, whilst $\varphi^*E$ may not coincide with $E$ relative to $g_{ab}$, it \textit{will} coincide with the equation one would get relative to $\varphi^*g_{ab}$. We will now outline one way of making this idea precise.

Let $B\overset{\pi}{\rightarrow}M$ be a fiber bundle. Following \citet[\S 6]{GerochPDEs}, a \textit{quotient bundle} of $B$ is a bundle $\Hat{B}\overset{\Hat{\pi}}{\rightarrow}M$ and a smooth map $\check{\pi}:B\rightarrow\Hat{B}$ such that $B\overset{\check{\pi}}{\rightarrow}B$ is a smooth fiber bundle and $\pi=\Hat{\pi}\circ\check{\pi}$. Note that this definition has as a consequence that if $\psi:M\rightarrow\Hat{B}$ is a (global) section of $\Hat{B}$, we can define a bundle $B_\psi\overset{\pi_\psi}{\rightarrow}M$ whose total space is $\check{\pi}^{-1}(\psi(M))$ and whose projection map is $\pi_\psi:=\pi_{|B_\psi}$.  This definition can be extended to natural bundles as follows.  Let $X:\mathcal{M}_n\rightarrow\mathcal{FB}$ be a natural bundle. A \emph{natural quotient} of $X$ is a tuple $(X, \check{\pi}, \Hat{X})$, where $\Hat{X}:\mathcal{M}_n\rightarrow\mathcal{FB}$ is a natural bundle and $\check{\pi}:X\rightarrow \Hat{X}$ is an epic natural transformation, such that for each object $M$ of $\mathcal{M}_n$, (1) the component $\check{\pi}_M$ is of the form $\check{\pi}_M=(\pi_{M},1_M)$, (2) $XM\overset{\pi_M}{\rightarrow}\Hat{X}M$ is a fiber bundle, and (3) $\pi_{XM}=\pi_{\Hat{X}M}\circ\pi_{M}$.


How does this help? Let $B\xrightarrow{\pi} M$ be a fiber bundle and suppose that $E$ a $k$th-order equation on $B$.  (Note that we are \emph{not} assuming $B$ or $E$ are natural -- to the contrary, the construction is of greatest interest when $E$ is not natural.)  Now suppose that $\mathfrak{B}:\mathcal{M}_n\rightarrow\mathcal{FB}$ is a natural bundle and  $(\mathfrak{E},\mathfrak{e})$ is an $\ell$th-order natural equation on $\mathfrak{B}$, for $\ell\geq k$.  Finally, suppose that $\mathfrak{B}$ admits a natural quotient $(\mathfrak{B},\check{\pi},\Hat{\mathfrak{B}})$; that for any $n$-manifold $N$ and (global) section $\phi: N\rightarrow \hat{\mathfrak{B}}N$, the restriction of $\mathfrak{E}$ to $J^k\Hat{\mathfrak{B}}N_\phi$ is an equation on $J^k\mathfrak{B}N_\phi$; and that that $\Hat{\mathfrak{B}}M$ admits a (global) section $\psi:M\rightarrow \Hat{\mathfrak{B}}M$, where $M$ is the base space of $B$, such that (i) $B\overset{\pi}{\rightarrow} M$ is isomorphic (as bundles) to $\mathfrak{B}M_\psi\overset{\Pi_\psi}{\rightarrow}M$, and (ii) the image of $E$ under that isomorphism coincides with the restriction of $\mathfrak{E}$ to $J^k\mathfrak{B}M_\psi$.  Then $\hat{\mathfrak{B}}$ can be thought of as a (natural) bundle of `background fields', different choices of which, on $M$, give rise to different equations on $B$ (including $E$) -- but where now we can also consider the `same' equation, by the lights of this naturalization, on bundles over other $n$-manifolds.  We will say that such a natural quotient $(\mathfrak{B},\check{\pi},\Hat{\mathfrak{B}})$ and natural equation $(\mathfrak{E},\mathfrak{e})$ together constitute a \emph{naturalization} of the equation $E$.  Basically, a naturalization is a natural equation that restricts, for a particular base space and choice of background fields, to the equation with which we began.  (Note that $\mathcal{E}$ may also determine equations on the background fields \citep[\S 6]{GerochPDEs}.)

Naturalization provides a way to say, in general terms, what it is for an equation to `depend' on some background structure, by thinking about the minimal structure which needs to be `added' to the equation in order to naturalize it. Let $B\overset{\pi}{\rightarrow}M$ be a fiber bundle, $E$ a $k$th-order equation on that bundle, and suppose that, for some natural quotient $(\mathfrak{B},\check{\pi},\Hat{\mathfrak{B}})$, we have a naturalization $(\mathfrak{E},\mathfrak{e})$ of $E$. 
In this case, we will say that $E$ depends \textit{only} on the background structure $\Hat{\mathfrak{B}}$, to $\ell$th order, where $\ell$ is the order of $(\mathfrak{E},\mathfrak{e})$. Of course, it might depend on less than this -- in general, one can always construct further naturalizations by adding `superfluous' background fields to the bundle $\hat{\mathfrak{B}}$ or by pulling back a naturalization to higher-order jet bundles of $\mathfrak{B}M$.  Note, too, that a naturalization does more than simply making the idea of dependence on background structure precise.  It also involves a choice of how to extend $E$ to bundles over different manifolds. For these reasons, one should not expect a general recipe for naturalizations.

We now return to Maxwell's equations on some manifold $M$, as described at the beginning of this section. Let $\mathfrak{EM}:\mathcal{M}_n\rightarrow\mathcal{FB}$ be the natural bundle whose fibers, over each $n$-manifold, are manifolds of pairs $(F_{ab},g_{ab})$, where $F_{ab}$ is a two-form and $g_{ab}$ is a Lorentzian metric.  We construct a natural quotient $(\mathfrak{EM},\check{\pi},\widehat{\mathfrak{EM}})$ by defining  $\widehat{\mathfrak{EM}}:\bar{\mathcal{M}}_n\rightarrow \mathcal{FB}$ as the natural bundle whose fibers, over each $n$-manifold, are the manifolds of Lorentzian metrics at each point $p$;\footnote{Recall that $\bar{\mathcal{M}}_n$ is the full subcategory of $\mathcal{M}_n$ whose objects are $n$-manifolds admitting a Lorentzian metic.} and taking $\check{\pi}$ to be the natural transformation that acts by taking pairs $(F_{ab},g_{ab})$ to metrics $g_{ab}$.  Sections of $\widehat{\mathfrak{EM}}$ can be thought of as manifolds with metric; our original spacetime can be thought of as corresponding to one such section. Now we define a natural equation $(\mathfrak{E},\mathfrak{e})$ on $\mathfrak{EM}$ by taking the source-free Maxwell's equation at every point relative to the metric and its Levi-Civita derivative operator at that point. The result is a naturalization of $E$, and $E$ depends only on the background fields $\widehat{\mathfrak{EM}}$ (i.e., the metric and derivative operator).

We said, above, that the existence of an $\ell$th-order naturalization $(\mathfrak{E},\mathfrak{e})$ of $E$ for some natural quotient bundle $(\mathfrak{B},\check{\pi},\Hat{\mathfrak{B}})$ can be used to capture the idea that $E$ depends only on the background structure $\Hat{\mathfrak{B}}$, to $\ell$th order. But we are also interested in when this $\ell$th-order dependence on the background structure $\Hat{\mathfrak{B}}$ exhausts the order of dependence on $\Hat{\mathfrak{B}}$ in $E$. This question is not trivial, for in general, the bundle $B$ might contain fields which encode information about the derivatives of (local sections of) $\Hat{\mathfrak{B}}$. For example, a second-order natural equation on the bundle of Lorentzian metrics can generically be re-expressed as a first-order system of equations on the bundle whose fibres are manifolds of pairs $(\nabla, g_{ab})$ of (torsion-free) derivative operators and Lorentzian metrics. For this reason, when considering the order of dependence of $E$ on $\Hat{\mathfrak{B}}$, it is necessary to impose a further condition on the naturalization $(\mathfrak{E},\mathfrak{e})$ to capture the idea that $E$ depends on $\Hat{\mathfrak{B}}$ to at most $\ell$th-order: namely, that for any object $N$ of $\mathcal{M}_n$, any sections $\psi, \chi:N\rightarrow \Hat{\mathfrak{B}}N$, and any $p\in N$, the condition $j^1_p\psi=j^1_p\chi$ implies $e_N(\mathfrak{E}N)_\psi\cap \pi_{J^1\mathfrak{B}}^{-1}(p)=e_N(\mathfrak{E}N)_\chi\cap \pi_{J^1\mathfrak{B}}^{-1}(p)$. This condition says, in some detail, that the $\ell$th-order dependence on the background structure $\Hat{\mathfrak{B}}$ exhausts the order of dependence on $\Hat{\mathfrak{B}}$ in $(\mathfrak{E},\mathfrak{e})$, in the sense that (local) sections of $\Hat{\mathfrak{B}}$ which are $\ell$-equivalent at some $p$ give rise to (`de-naturalized') equations which agree at $p$. 

\section{Minimal Coupling}\label{sec:MinCoup}
In the previous section, we showed how a system of equations on a particular bundle may be extended to a natural equation, via what we called `naturalization'.  We will now use naturalization to give a precise mathematical statement of when a system of equation is \emph{minimally coupled}. As usually understood, minimal coupling is a heuristic used to construct equations for matter fields in general relativity -- though really, it can refer to any of a family of (\textit{a priori} distinct) heuristics. For example, `minimal coupling' is often identified as the prescription for constructing equations in general relativity from equations in special relativity whereby one replaces all instances of the Minkowski metric $\eta_{ab}$ with $g_{ab}$, and all instances of the Levi-Civita connection for $\eta_{ab}$ with the Levi-Civita connection for $g_{ab}$. As has been noted by \citet[ch.~6.2]{Trautman1965} and \citet[\S 16.3]{Misner+etal}, this procedure is apparently ambiguous (or as \citet{Fletcher2020} puts it, `hyperintensional') for systems of second-order or higher; it also fails to yield second-order systems that are independent of curvature terms.\footnote{In other words, there is a conflict between the above `minimal coupling' prescription and an unqualified version of the statement of minimal coupling we consider below which would apply to systems of equations of arbitrary order, which poses a problem for taking the former as an explication of the latter.} As such, we will not pursue this approach here. 

Instead, we take the idea of minimal coupling to be the following:
\begin{quote}
    Minimally coupled first-order systems of equations depend only on the metric and its first derivatives. 
\end{quote}
This heuristic is often taken as motivation for the above `minimal coupling' prescription. Of course, as it stands, this sense of minimal coupling is also imprecise -- not least because we have not yet said what the relevant notion of dependence is.\footnote{Often, this is implicitly cashed out as the non-appearance of curvature terms in the syntactic expression of a system of equations. We take this to be a non-starter for similar reasons to those given above: non-appearance of curvature terms in the syntactic expression of a system of first-order equations is not preserved under all syntactic manipulations of that system of equations which preserve its space of solutions. cf.~\citet{WeatherallDogmas} and \citet{FletcherWeatherallpart2}.} It turns out naturalization can help, since naturalization makes precise the idea that a system of equations depends (only) on some background structure.  

Let $G:\bar{\mathcal{M}}_n\rightarrow\mathcal{FB}$ be the (first-order, natural) bundle whose fibers are manifolds of Lorentzian metrics.\footnote{This is the bundle we previously called $\widehat{\mathfrak{EM}}$.}
We will say that a first-order natural theory $(X,E,e)$ is \emph{minimally coupled} if there exists a natural quotient $(X,\check{\pi},G)$ of $X$ (over $G$) such that, for any object $N$ of $\bar{\mathcal{M}}_n$, any sections $\psi, \chi:N\rightarrow GN$, and any $p\in N$, the condition $j^1_p\psi=j^1_p\chi$ implies $e_N(\mathfrak{E}N)_\psi\cap \pi_{J^1F}^{-1}(p)=e_N(\mathfrak{E}N)_\chi\cap \pi_{J^1F}^{-1}(p)$.  Now let $X:\bar{\mathcal{M}}_n\rightarrow\mathcal{FB}$ be a natural bundle, let $M$ be an object of $\bar{\mathcal{M}}_n$, and let $E\subset J^1XM\rightarrow XM\rightarrow M$ be a first-order equation on $XM$.  A \emph{minimal coupling} for $E$ is a minimally coupled naturalization $(\mathcal{X},\mathcal{E},e)$ over a natural quotient $(\mathcal{X},\check{\pi},G)$.


There are several features of this explication of minimal coupling worth commenting on further. First, minimal coupling gives us a recipe for naturalizing first-order flat spacetime equations on natural bundles.  A complete discussion of this point requires further machinery, and we will postpone it to future work, but one can see the point informally as follows:  since equations are defined locally in the present context, as submanifolds of an appropriate jet bundle, and since jet bundles over natural bundles are natural, a first-order equation in flat spacetime, depending only on the metric, will determine a unique equation for curved spacetimes via a pushforward construction.\footnote{It is is important to note that at a point, every curved metric agrees to first order with a flat metric \citep{FletcherWeatherallpart1}, and so one can determine the equation in curved spacetime fiberwise via pushforwards along diffeomorphisms that realize that first-order agreement.}  In fact, it will turn out that in the cases where minimal coupling can be applied at all, it will determine a unique naturalization, up to natural isomorphism.  In this sense, minimal coupling, as we define it here, resolves the hyperintentionality and ambiguity concerns raised by other authors.

Second, we highlight that our definition applies only to first-order systems.  This is because if naturalization is the right explication of what it is for an equation to depend on some background structure, then it can make sense to demand that a system of equations depends only on the metric and its first derivatives only when that system is itself first-order.\footnote{This is because, given a $k$th-order ($k>1$) equation $E$ on a natural bundle $X:\mathcal{M}_n\rightarrow\mathcal{FB}$, and a $k$th order naturalization $(\mathfrak{E},\mathfrak{g})$ of $E$ over $G$, there is no way to say that $\mathfrak{E}$ is `genuinely first-order' in the metric whilst being `genuinely $k$th-order' in the physical fields. (Of course, we can say what it is for $\mathfrak{E}$ to be genuinely first-order in \textit{both} the metric and physical fields, i.e., when $E$ coincides with the pullback of some first-order natural equation on $X$ via the map $j^k_1\pi$; we can also say what it is for an equation to be `genuinely first-order' in the physical fields whilst being `genuinely $k$th-order' in the metric, i.e., when $\mathfrak{E}$ is a naturalization of a first-order system of equations on $Y$ which is not the pullback of any lower-order equation on $X$ via some $j^k_l\pi$.)} This gives us a clearer understanding of the sense in which the scope of minimal coupling needs to be `restricted' to first-order systems (or: is `ambiguous'/`hyperintensional' for higher-order systems): first-order dependence on the metric only picks out an `intrinsic' property of a system of equations -- viz., the existence of a naturalization over the bundle of Lorentzian metrics -- when that system of equations is first-order. On the other hand, one also sees a sense in which the scope of minimal coupling is broader than this restriction might appear to suggest, in that it can be applied to \textit{any} system of equations which can be transformed into an equivalent first-order system via an appropriate choice of auxiliary variables. Any ambiguities must be resolved at the stage of identifying these first-order equations.  Lingering curvature terms then just mean the equation is not minimally coupled.

To illustrate this point, consider how our framework can be used to resolve questions about the `correct' application of minimal coupling in e.g.~electromagnetism. Let $(M, g_{ab})$ be a relativistic spacetime, and consider the source-free Maxwell's equation, this time expressed in terms of a vector potential. In the literature, at least two candidates for this equation have been identified (see, e.g.~\citet[\S 16.3]{Misner+etal}):
\begin{align}
    2\nabla_n\nabla^{[n}A^{a]}&=0 \label{eq:MCMaxwell}\\
    2\nabla_n\nabla^{[n}A^{a]}-R\indices{^a_n}A^n&=0.\label{eq:notMCMaxwell}
\end{align}
It is sometimes claimed that both \eqref{eq:MCMaxwell} and \eqref{eq:notMCMaxwell} are minimally coupled, because both reduce to the special relativistic equations.  Nonetheless, \eqref{eq:MCMaxwell} is taken to be the `correct' equation. The standard rationale for this is that \eqref{eq:MCMaxwell} (but not \eqref{eq:notMCMaxwell}) agrees with Faraday tensor formulation of Maxwell's equations upon making the identification $F_{ab}=2\nabla_{[a}A_{b]}$, or perhaps that \eqref{eq:MCMaxwell} (but not \eqref{eq:notMCMaxwell}) is gauge-invariant (see also \citet[p.~390]{Misner+etal} who argue that ``[c]oupling to curvature surely cannot occur without some physical reason'').

Our explication of minimal coupling provides a somewhat different resolution to this problem.  In fact, only equation \eqref{eq:MCMaxwell} is minimally coupled. Equation \eqref{eq:notMCMaxwell} is not. To see this, consider first the equation $2\nabla_n\nabla^{[n}A^{a]}=0$. We can express this as a first-order system by introducing the Faraday tensor as an auxiliary variable:
\begin{align}
    2\nabla_{[a}A_{b]}&=F_{ab} \label{eq:Faraday}\\
    \nabla_{[a}F_{bc]}&=0 \label{eq:Maxwell1}\\
    \nabla_nF^{na}&=0. \label{eq:Maxwell2}
\end{align}
To translate this into our framework, let $X:\bar{\mathcal{M}}_n\rightarrow\mathcal{FB}$ be the (natural) bundle whose fibers are manifolds of pairs $(A_a, F_{ab})$, where $A_a$ is a one-form and $F_{ab}$ a two-form, let $M$ be any object in $\bar{\mathcal{M}}_n$, $g_{ab}$ a Lorentzian metric on $M$, and let $E\subset J^1XM\rightarrow XM\rightarrow (M, g_{ab})$ be the subbundle of $J^1XM$ determined by \eqref{eq:Faraday}-\eqref{eq:Maxwell2}. There is an obvious first-order naturalization of $E$ over the bundle of Lorentzian metrics, which is constructed in exactly the same way as for the Maxwell system in section \ref{sec:naturalization}: we let $F:\bar{\mathcal{M}}_n\rightarrow\mathcal{FB}$ be the natural bundle whose fibers are manifolds of tuples $(A_a, F_{ab}, g_{ab})$, and $\mathfrak{E}\subset J^1FM$ be the equation obtained by taking \eqref{eq:Faraday}-\eqref{eq:Maxwell2} at every point relative to $g_{ab}$ and its first derivatives at that point. (Note that here we have made essential use of the fact that \eqref{eq:Faraday} does not depend on $g_{ab}$, in the sense of section \ref{sec:naturalization}.)

On the other hand, consider the equation $2\nabla_n\nabla^{[n}A^{a]}-R\indices{^a_n}A^n=0$. Again, we can express this as a first-order system by introducing the Faraday tensor as an auxiliary variable: equations \eqref{eq:Faraday} and \eqref{eq:Maxwell1} remain unchanged, but in place of \eqref{eq:Maxwell2} we have
\begin{align}
    \nabla_nF^{na}-R\indices{^a_n}A^n&=0.\label{eq:notMaxwell}
\end{align}
This equation involves explicit second-order dependence on $g_{ab}$ and therefore cannot be naturalized by a first-order natural equation on $F$.  Hence it cannot be minimally coupled.

\section{Symmetries of a Theory}\label{sec:Earman}

The framework presented thus far provides a useful new perspective on another vexed question in philosophy of physics, concerning how to make precise the idea of `dynamical symmetries' or, more generally, `symmetries of  a theory'.  We will not review the large literature on symmetries and their significance.  Instead, our starting point will be to note that there is a standing difficulty associated with giving a precise mathematical definition of the symmetries of a physical theory that can support the myriad inferences philosophers wish it to support.\footnote{\citet{Belot} has made this point especially forcefully.}  

The key difficulty concerns how to make sense of a transformation that preserves a system of equations.  One strategy for capturing this idea, sometimes called the `syntactic approach', considers coordinate transformations that preserve the `form' of an equation \citep{WEST, Brown}.  But as several authors have argued, without an adequate theory of equations, it is not clear when two equations are the same equation written in different forms or when two equations in the same form are in fact distinct \citep[see, e.g.][]{WeatherallDogmas, FletcherWeatherallpart2}.  Another strategy is semantic: one defines symmetries as transformations that take models of a theory to models of a theory (or, solutions to solutions).  But again, many authors have observed that on its own, this approach is inadequate, because arbitrary permutations of models are not symmetries, and it is not clear how to identify, in a non-question-begging way, what else needs to be preserved \citep{Ismael+vanFraassen,Belot}.

Here we explore what can fruitfully be said about symmetries in the context of natural theories, and show how symmetry-based reasoning in this framework does recover some key claims from the philosophy of symmetries literature (at least for the theories that can be described in this framework). Most importantly, the present approach can explain the cogency of Earman's famous symmetry matching principles, (SP1) and (SP2), discussed below.  The strategy is to identify the geometric structure of a natural theory as what needs to be `preserved' by a symmetry transformation.  This approach is related to the syntactic approach just mentioned, in the sense that we take an equation, or system of equations, as a way of defining a certain geometric structure whose invariance properties we care about; and it is related to the semantic approach in the sense that the structure in question is not merely a mathematical expression, but rather a geometrical representation of the (local) `solutions' to a system of equations.

We will give three different definitions of symmetry of a theory, all of which are salient but capture somewhat different ideas.  First, suppose that we have a $k$th order natural theory $(X,E,e)$.  Then there is a certain sense in which every diffeomorphism $\chi:M\rightarrow N$ between $n$-manifolds can be thought of as a symmetry of the theory, since every diffeomorphism lifts to an isomorphism between $J^kXM$ and $J^kXN$ that maps $EM$ into $EN$.  This means, in particular, that every diffeomorphism maps `solutions into solutions' in the sense that if $\varphi:M\rightarrow XM$ solves $EM$, then $\chi_*\circ\varphi\circ \chi^{-1}:N\rightarrow XN$ is a solution to $EN$.  Indeed, natural theories are constructed in order to make this true.  It captures a sense in which such theories can be said to be `diffeomorphism invariant'.  This is one way of substantiating the claim that natural theories are `generally covariant'.

Notice that \textit{every} natural theory is diffeomorphism invariant in the present sense, including, for instance, naturalized (minimally-coupled) Maxwell theory, naturalized (minimally-coupled) Klein-Gordon theory, and so on.  But often a distinction is drawn between a theory such as general relativity, which is said to have all diffeomorphisms as symmetries, and a theory like Maxwell's theory, which is said to have only isometries as symmetries (or, say, in the special case of Maxwell's theory on Minkowski spacetime, Poincar\'{e} transformations as symmetries).  This sense of symmetry is more specialized.  To recover it, suppose we have a natural quotient $(B,\check{\pi},\hat{B})$ and a natural equation $(E,e)$ on $B$ that restricts, for each $n-$manifold $M$ and section $\psi$ of $\widehat{BM}$, to an equation on $BM_{|\psi}$. (Such a structure may arise due to naturalization---though it may also just be considered on its own.) Fix a section $\psi:M\rightarrow \hat{B}$.\footnote{If $(E,e)$ induces an equation on $\hat{B}$, we assume $\psi$ satisfies those equations.}  Now consider the pullback of $EM\subseteq J^kBM\rightarrow BM\rightarrow \hat{B}M\rightarrow M$ to $\psi[M]\subseteq \hat{B}$ along the inclusion map, i.e., the restriction of our equation to the image of the section $\psi$ understood as a submanifold of $\hat{B}$.  We can think of this as a `de-naturalized' equation: it is the natural equation $(E,e)$ for a specific choice of background structure.  

We now introduce two definitions, relativized to this data.  A \emph{spacetime symmetry} is a diffeomorphism $\varphi:M\rightarrow M$ whose lift $\varphi_*$ to $\psi[M]\rightarrow M$ is a bundle isomorphism (or, equivalently, whose lift to $\hat{B}\rightarrow M$ satisfies $\varphi_*\circ\psi=\psi\circ\varphi$).  This is a diffeomorphism that preserves the particular choice of background structure picked out by $\psi$.  To take an example, consider naturalized Maxwell theory.  Suppose $M$ is $\mathbb{R}^4$ and $\psi$ is the section of the bundle of metrics corresponding to Minkowski spacetime.   Then the spacetime symmetries would be those diffeormorphisms of $M$ that preserve the metric -- namely, the Poincar\'e transformations.  More generally, in this case, the spacetime symmetries would be the isometries of a given metric.  Similarly, if the background structure needed to naturalize an equation consists in something other than a Lorentzian metric, the spacetime symmetries would be the diffeomorphisms that preserve particular sections.  Note that spacetime symmetries do not need to preserve \emph{all} sections -- only trivial maps would do that.  Note, too, that in the case where the natural quotient and and natural equation arise from a naturalization, we do not begin by specifying what the spacetime structure is.  We get that from the equation itself, by investigating what background structure is needed to naturalize it.\footnote{There is a further question of whether all background structure should could as \emph{spacetime} structure.  We do not address that question here.  For minimally coupled systems, it seems clear that the background structure is spacetime structure.}  

A \emph{dynamical symmetry}, meanwhile, is an automorphism $\varphi:M\rightarrow M$ whose lift to $J^kB_{|\psi[M]}\rightarrow B_{|\psi[M]}\rightarrow M$ preserves $E$.  In other words, a dynamical symmetry is a diffeomorphism from the base space to itself that preserves the equation, \emph{for some particular choice of background structure}, but which need not preserve the background structure itself.  Here is an example.  Consider, again, the case of the (source-free) naturalized (minimally-coupled) Maxwell equations, and let $\varphi$ be a diffeomorphism that implements a constant conformal transformation, i.e., one that takes the metric $\eta_{ab}$ to $\Omega\eta_{ab}$ for some number $\Omega$, constant on spacetime.  By construction, this map is not a spacetime symmetry.  But it nonetheless preserves the vacuum Maxwell equations.  More generally, the idea is that a dynamical symmetry is one that preserves the equation (and, by consequence, its solutions), irrespective of whether it happens to preserve the spacetime structure.  It is important, however, that the sense of `equation' being preserved here is already relativized to a choice of background structure, since otherwise \emph{every} diffeomorphism would be a dynamical symmetry (for some background or other).

These two definitions of symmetry clarify Earman's famous symmetry matching principles, (SP1) and (SP2), which state that for some theory $T$ \citep[p. 46]{WEST}.  
\begin{quote}
\begin{enumerate}[SP1]
\item Any dynamical symmetry of $T$ is a space-time symmetry of $T$.
\item Any space-time symmetry of $T$ is a dynamical symmetry of $T$.
\end{enumerate}
\end{quote}
Both of these principles are well-defined in the present context, at least for natural theories, and have clear application -- with the caveat that one should also quantify over all possible sections of the background bundle $\hat{B}$.  But we also see a kind of asymmetry between the principles, reflected in some debates about the status of (SP1).  As we have set things up, (SP2) holds trivially.  Every diffeomorphim that preserves background structure in the sense we have described automatically preserves the equation.  This is a direct consequence of naturality.  The issue is not that (SP2) could never fail; rather, it is that naturalization enforces (SP2), and so a theory that violates (SP2) is one that is not (yet) natural.  If one restricts attention to natural(ized) theories, then, (SP2) is guaranteed.

But (SP1) has a different status.  In a sense, (SP1) is a diagnostic for \textit{over} naturalization.  As we noted above, naturalization is not a unique prescription.  In principle, one could build too much structure into a background bundle -- structure that does not need to be preserved by a diffeomorphism to fix the equations, and thus structure that the equations do not depend on.  A failure of (SP1) signals that a different background bundle could be chosen to naturalize the equation with which one began.  Satisfaction of (SP1), meanwhile, suggests that the background bundle is `minimal', in the sense that there is no way to preserve the equation without also preserving the background structure.

Thus we see the difference between (SP1) and (SP2), which concerns the status of background structure in naturalization.  A failure of (SP2) indicates that the equation is not well-behaved under the action of diffeomorphisms.  It is not natural, in the sense that it has some unacknowledged dependence on structure that is neither built out of the manifold nor reflected in the background bundle.  Or in other words, it is not generally covariant.  This strikes us as a bad failure, because it means the equation is in some sense ambiguous.  A failure of (SP1), by contrast, does not introduce serious problems.  It violates occamist norms, perhaps, but there may nonetheless be pragmatic reasons to prefer theories that violate (SP1).  Take, for instance, the example given above, of naturalized Maxwell theory.  As we saw, in the case of the vaccuum theory, there is an extra dynamical symmetry that is not a spacetime symmetry (assuming one takes the background structure to be a metric and derivative operator).  But this symmetry generally goes away once once couples the theory to sources, and there are good reasons to presuppose the background structure needed for the full theory even when restricting attention to the vaccuum sector.  

That said, (SP1) plausibly has more normative force when one assumes that the equation $E$ contains \emph{all} fundamental physical laws, understood as a large system of coupled equations.  In that case, the presence of a dynamical symmetry that is not a spacetime symmetry suggests that the background bundle includes structure that could not matter for the dynamics of any physical system at all, even in principle.  And as \citet{Roberts} observes, this in turn means that that background structure can have no empirical significance, since no physical phenomenon depends on it.  One might then move to a weaker spacetime structure.  But that is not forced on us by naturality.

We now turn to the third notion of a `symmetry of a theory'.  This one begins in a somewhat different place.  The idea is that we have a given a definition of a `natural theory' as a precise mathematical structure consisting of a natural bundle and a natural equation on that bundle, i.e., two functors and a natural transformation.  This suggests the following definition: a \emph{natural theory isomorphism} $(\psi,\Psi):N\rightarrow \tilde{N}$ between natural theories $N=(X,E,e)$ and $\tilde{N}=(\tilde{X},\tilde{E},\tilde{e})$ of order $k$ on manifolds of dimension $n$ is a pair of natural isomorphisms $\psi:X\rightarrow\tilde{X}$ and $\Psi:E\rightarrow \tilde{E}$ such that $\psi_*\circ e=\tilde{e}\circ\Psi$, where $\psi_*:J^k X\rightarrow J^k\tilde{X}$ is the natural isomorphism whose factor at each object $M$ of $\mathcal{M}_n$ is the pushforward of $k$-jets along $\psi_M$.\footnote{An analogous definition could be given for homomorphisms of these structures, by substituting natural transformations for natural isomorphisms.  Note that our definition requires the equations be of the same order, since that is what we will need in what follows.  But one might consider a more general notion of equivalence that allows comparison of equations of different order, to capture the idea that equations of a given order imply equations at higher orders (and that generally equations at one order can be rewritten as equations at lower orders).  We do not pursue that topic here, but suggest it would be an interesting extension.}
We can then say that a \emph{symmetry of a natural theory} is an automorphism of that theory, i.e., an isomorphism from that theory to itself.  

In our view, this definition is `correct' from a mathematical perspective, and because it captures the desideratum that a symmetry of a theory should be some transformation that preserves the structure of the theory.  In that sense it is a principled definition to consider.  But it is also striking, because it does not capture the sorts of things philosophers of physics usually have in mind when they talk about symmetry.  Indeed, most physically important natural theories, including general relativity and naturalized (minimally-coupled) Maxwell theory, apparently have \emph{no} non-trivial symmetries.\footnote{Something similar can be said for \emph{gauge natural theories}, including non-Abelian Yang-Mills theory, though we do not discuss those here.  That said, electromagnetism, conceived as an Abelian Yang-Mills theory with structure group $U(1)$ does have a symmetry: namely, a global $U(1)$ symmetry, which rotates every fiber of every bundle by an element of $U(1)$.}  This is true even when we formulate Maxwell's theory using vector potentials.  There seems to be no interesting sense in which standardly cited examples of `symmetries' of well-known theories, such as diffeomorphisms, gauge transformations, conformal transformations, Poincar\'e transformations, and so on, are or give rise to symmetries of natural theories.  Whatever interest those classes of transformations may hold, it is not because they preserve the structure of the theory.  To the contrary, the role they appear to play in analysis of theories is better tracked by the senses of symmetry already discussed in this section, which are not really symmetries of the theory, per se, but diagnostics for establishing naturality and structural dependencies of various sorts within a theory.

As a final remark, we wish to contrast what we have said here with a different claim that one sometimes sees regarding general relativity, which is that the `diffeomorphism group' is the `symmetry group' of the theory. On our view, the most direct reading of that claim is that general relativity can be seen is a mathematical object whose automorphism group coincides with the `diffeomorphism group'.  (We use scare quotes, because of course there is no single diffeomorphism group -- rather, there are diffeomorphism groups for each manifold.) We have not recovered that claim, and we do not see any sense in which it is true.  Indeed, we would suggest that there is a straightforward sense, given by the notion of symmetry of a theory just introduced, in which the symmetry group of general relativity is the trivial group: as a theory, it has no symmetries at all.  That said, what we \emph{do} find is that given a fixed $n$-manifold $M$, the functor $E$ (faithfully) maps the diffeomorphism group of $M$ into the automorphisms of $EM$, and it does so in a sufficiently natural way that we can think of the equation as being `preserved' under the action of diffeomorphisms.  More generally, the functor $E$ faithfully maps the entire `diffeomorphism groupoid', i.e., the groupoid consisting of just the objects and isomorphisms of $\mathcal{M}_n$, into the `equation groupoid' $E[\mathcal{M}_n]$.  Thus, while diffeomorphisms are not automorphisms of the theory conceived as a single mathematical object, nonetheless each diffeomorphism, including automorphisms, can be seen as a `local' symmetry of the theory, i.e., an isomorphism of the theory localized to manifolds related by that diffeomorphism.

\section{The Hole Argument, Naturally}\label{sec:HoleArg}

Finally, we turn to the solution properties of natural theories.  We are particularly focused on a tension between two desiderata one might impose on a physically reasonable theory.  The first desideratum is that it be natural, in the sense we have discussed here.  The second desideratum is that the theory be deterministic.  Of course, what it means to say that a theory is deterministic is itself a matter of significant philosophical dispute, but in the context of partial differential equations there is a standard answer that seems like a plausible starting point.   We will say that a system of equations on an $n$-manifold has a \emph{well-posed initial value problem} if, given suitable `initial data' on an $(n-1)$-submanifold, there exists an open set $O$ containing that submanifold and a smooth solution on $O$, agreeing with the initial data on the submanifold, which is unique in the sense that any other such solution agreeing with the initial data coincides with that solution on $O$.  (Precise definitions are available, here, but require preliminaries that would take us too far afield.  The proposition we presently prove will not depend on these details.  See \citet{JOhn} or \citet{GerochPDEs} for an extended discussion.)  The idea is that an equation is well-posed if data about a solution on a surface -- i.e., `at a time' -- uniquely determines that solution off the surface, at least on some open set around the initial data.  

Let us say that a natural theory $E$ on a natural bundle $X:\mathcal{M}_n\rightarrow \mathcal{FB}$  is `sufficiently rich' if it admits solutions $\varphi:M\rightarrow XM$ for which there exist diffeomorphisms $\chi:M\rightarrow M$ such that $\chi_*\circ \varphi \neq \varphi\circ \chi$.  To say that a theory is sufficiently rich means that it admits solutions that are non-constant, in the sense that they vary from point to point by some standard of comparing points.  This condition is very weak.  The theory of constant real scalar fields is not sufficiently rich, but any other case of plausible interest is.  Then we have the following result.

\begin{thm}No sufficiently rich natural theory has a well-posed initial value problem.\end{thm}
\noindent Proof. Let $(E,X,e)$ be a natural theory, and consider any smooth $n$-manifold $M$.  Let $\varphi:M\rightarrow XM$ be a solution and let $\chi:M\rightarrow M$ be any diffeomorphism such that (1) $\chi$ acts as the identity outside some region $O$ of compact closure and (2) within $O$,  $\chi_*\circ \varphi \neq \varphi\circ \chi$.  (That such a diffeomorphism can be found is a consequence of sufficient richness.)  Now suppose $U$ is an initial segment that does not intersect $O$.  It follows that $\varphi'=\chi_*\circ\varphi\chi^{-1}$ is a solution for which $(1_M)_*\circ\varphi_{|U}=\varphi'\circ (1_M)_{|U}$ but $(1_M)_*\circ\varphi \neq \varphi'\circ (1_M)$. $\square$

The proof of this theorem should seem familiar: it is just the hole argument.  The point is that naturality and sufficient richness are sufficient for a version of the hole argument to go through.  (One might add that naturality and sufficient richness are also \emph{necessary} for the standard hole argument, since that argument takes for granted that the pushforward of fields is both well-defined and takes solutions to solutions.)

How does this result square with the fact that many systems of equations in physics \textit{do} have well-posted initial value problems?  Take, again, our running example: Maxwell's equations.  As we saw above, this system is not natural, and so the theorem does not immediately apply to it.  That makes sense, since Maxwell's equations are symmetric hyperbolic and linear, so they do admit a well-posed initial value problem.  On the other hand, the system can be naturalized, as discussed above.  And the theorem \textit{does} apply to the naturalized theory.  In order to recover the symmetric hyperbolic character of the system, one needs to fix a choice of background structure, by fixing a section of the bundle of background fields.  In other words, Maxwell's equations are well-posed only relative to a choice of metric.  This situation turns out to be completely general, as established by the proposition.

This result should give a different perspective on the hole argument, well-posedness, and determinism.  A system of partial differential equations on a manifold can have a well-posed initial value problem only if the equations depend on some background structure.\footnote{Note that referring to background structure is necessary here but not sufficient -- in fact, the rigidity properties of the background structure matter.  But we postpone a discussion of that topic for future work.}  Once we naturalize that background structure, well-posedness is ruined.  This is the tension we mentioned above, between well-posedness and naturality.  Well-posedness is not a `natural' condition, in the sense that no natural theory can satisfy it.

On the other hand, there are clearly senses in which natural theories \textit{can} be deterministic, though the senses of determinism is not captured by well-posedness. Instead, what we want is something like `well-posedness up to the action of diffeomorphism', which would be a condition where given initial data on some suitable submanifold, one requires existence on an open set containing that submanifold that is unique up to (unique) diffeomorphism \citep[c.f.][]{Halvorson+etal}.  This is the sort of determinism one finds for Einstein's equation, the best-studied example of a physically relevant natural theory.\footnote{The naturalized Maxwell equations are not even deterministic in this sense, because they cannot specify the behavior of the background fields.  But the Einstein-Maxwell system, which imposes Einstein's equation on the background structure and couples it to the Maxwell stress-energy tensor, would be deterministic in this sense.}  This, in turn, suggests that the tension is not between naturality and determinism, \textit{per se}, but naturality and a particular way of making determinism precise.  

From this perspective, the hole argument as it is usually deployed in the foundations of physics literature, can be seen as highlighting a salient mathematical difference between Einstein's equation and other, non-natural but frequently encountered equations, such as Maxwell's equations.  Formalist responses to the hole argument \citep{Mundy,WeatherallHoleArg, Shulman, Ladyman+Presnell, Bradley+Weatherall}, meanwhile, can be viewed as insisting that well-formed mathematical claims must be suitably natural -- and therefore that deploying something like well-posedness in the context of a natural equation involves a subtle mathematical error, a mismatch between the sort of equation under consideration and the compatible uniqueness properties of solutions.  This discussion suggests a different way forward, which is to identify more precisely what should be meant by `deterministic' for a natural theory, analogous to well-posedness but not requiring background structure, and to develop a mathematical theory of when a natural theory will be deterministic in this way. 

\section{Conclusion}\label{sec:conclusion}

We have introduced the concept of a \emph{natural equation}, which is a (natural) system of differential equations on a natural bundle, and we have suggested that this proposal completes the project initiated by \citet{March+Weatherall} to give a precise definition of a generally covariant (classical) field theory.  We have shown that whilst not all equations in classical field theory are natural, many can be \emph{naturalized}, i.e., turned into natural equations by identifying implicit background structure on which they depend, and we applied our notion of naturalization to explicate minimal coupling.  We then applied this formalism to distinguish several precise senses in which a theory might be said to have symmetries; and we used those distinctions to clarify Earman's famous symmetry matching principles.  Finally, we stated and proved a fundamental theorem showing that no natural theory can have a well-posed initial value problem.  This makes precise a common claim that generally covariant theories will generically be subject to hole argument-type constructions; but it also shows that the standard way of setting up well-posedness for hyperbolic partial differential equations is incompatible with naturality.

\appendix

\section{}\label{app:}

\begin{thm}
    Let $X:\mathcal{M}_n\rightarrow\mathcal{FM}$ be a natural bundle, with $k$th jet bundle $J^kX$, and let $M\in\mathrm{ob}(\mathcal{M}_n)$. Let $E$ be any non-empty closed embedded submanifold of $J^kXM$, and suppose that for any diffeomorphism $\varphi:M\rightarrow M$, $\varphi^k(E)=E$, where $\varphi^k(j^k_p\psi)=j^k_{\varphi(p)}\varphi^*(\psi(\varphi^{-1}(p)))$ for any local section $\psi:U\rightarrow XM$ and $p\in U$ and $(\varphi^*,\varphi)=X\varphi$. Let $\pi_{E}: E\rightarrow M$ denote the restriction of $\pi_{J^kXM}:J^kXM\rightarrow M$ to $E$. Further suppose that $M$ is connected and without boundary. Then $\pi_{E}: E\rightarrow M$ is a subbundle of $\pi_{J^kXM}:J^kXM\rightarrow M$. 
\end{thm}
\begin{proof}
    We show this in three parts:
    \begin{itemize}
        \item[(i)] $\pi_{E}:E\rightarrow M$ is surjective.
        \item[(ii)] $\mathrm{d}\pi_{E}:T_{j^k_p\psi}E\rightarrow T_pM$ is (pointwise) surjective.
        \item[(iii)] $\pi_{E}:E\rightarrow M$ is locally trivial.
    \end{itemize}

    For (i), suppose otherwise, i.e.~there is some $p\in M$ such that $p\notin\pi(E)$. But now consider any $j^k_q\psi\in E$ and any diffeomorphism $\varphi:M\rightarrow M$ with the property that $\varphi(q)=p$. (That such a $j^k_q\psi$ and $\varphi$ exist follows from the fact that $E$ is non-empty and $\mathrm{diff}(M)$ is transitive on $M$ for connected manifolds without boundary.) Then $\varphi^k(j^k_q\psi)=j^k_{\varphi(q)}\varphi^*(\psi(\varphi^{-1}(q)))=j^k_{p}\varphi^*(\psi(\varphi^{-1}(q)))\notin E$. But $\varphi^k(E)=E$, so $\pi_{E}:E\rightarrow M$ must be surjective.

    For (ii), consider any $p\in M$ and any $\xi_p\in T_pM$. We can always construct a smooth complete vector field $\xi$ on $M$ extending $\xi_p$ as follows: let $\xi'$ be any smooth vector field extending $\xi_p$, let $f:M\rightarrow \mathbb{R}$ be any smooth scalar field with support on some compact region $U\subset M$, $p\in U$ and satisfying $f(p)=1$, and let $\xi=f\xi'$. (That $\xi$ is complete is then a consequence of the fact that it is compactly supported.) Let $\gamma_q(t)$, $t\in\mathbb{R}$ denote the integral curves of $\xi$, where $\gamma_q(0)=q$, $q\in U$, and let $\varphi_t:M\rightarrow M$, $\varphi_t(q)=\gamma_q(t)$ denote the $1$-parameter family of diffeomorphisms generated by $\xi$. In turn, this induces a $1$-parameter family of diffeomorphisms $\varphi^k_t:J^kXM\rightarrow J^kXM$, generated by the vector field $\xi^k$, $\xi^k(j^k_q\psi)={\frac{\partial}{\partial t}|}_{t=0}\varphi^k_t(j^k_q\psi)\in T_{j^k_q\psi}J^kXM$. (That the lift $\varphi^k_t$ of $\varphi_t$ to $J^kXM$ is again a 1-parameter family of diffeomorphisms is a consequence of naturality, see \citet[\S20]{Kolar+etal} for details.) Since $\varphi^k_t(E)=E$, we know that $\xi^k(j^k_q\psi)\in T_{j^k_q\psi}E$ whenever $j^k_q\psi\in E$.

    But now (letting $\gamma^k$ denote the integral curves of $\xi^k$): $\mathrm{d}\pi_{E}({\frac{\partial}{\partial t}|}_{t=0}\varphi^k_t(j^k_p\psi))= \mathrm{d}\pi_{E}([\gamma^k_{j^k_p\psi}(t)]_{t=0}=[\pi_{E}\circ \gamma^k_{j^k_p\psi}(t)]_{t=0}=[\gamma_p(t)]_{t=0}=\xi_p$. Since $\xi_p$ was arbitrary, it follows that $\mathrm{d}\pi_{E}:T_{j^k_p\psi}E\rightarrow T_pM$ is (pointwise) surjective.

    Finally, we establish (iii). First, we show that the fibers of $\pi_{E}:E\rightarrow M$ at any two points in $M$ are isomorphic. Let $p\in M$ and choose any other $q\in M$. Again, using that $\mathrm{diff}(M)$ acts transitively on $M$, there is some diffeomorphism $\varphi:M\rightarrow M$ such that $\varphi(p)=q$. But $\varphi^k(E)=E$, and we know that for any $j^k_x\psi\in E$, $\pi_{E}\circ \varphi^k(j^k_x\psi)=q$ iff $x=p$, so $\varphi^k_{|\pi_{E}^{-1}(p)}:\pi_{E}^{-1}(p)\rightarrow \pi_{E}^{-1}(q)$ is a diffeomorphism.

    Now let $p\in M$, $(U, \phi)$ be a local chart with $\phi(p)=0$. We will show that there exists a local trivialisation of $\pi_{E}:E\rightarrow M$ on some subneighbourhood $U'\subset U$. For this, first let $B\subset M$ be an open ball about $0$ in $(U, \phi)$. For each $q\in B$, we have a smooth vector field $\xi_q$ (unique up to a multiplicative constant) on $B$ which is constant with respect to $(U, \phi)$, and which has a (directed) integral curve from $p$ to $q$. We can extend the $\xi_q$ to (complete) vector fields on $M$ with support on $B$. Each family of integral curves $\gamma_q$ of $\xi_q$ generates a 1-parameter group of diffeomorphisms $\varphi_t:M\rightarrow M$; in particular, for each $q$, we have some $t$ such that $\varphi_t(p)=\gamma_p(t)=q$. Denote this $\varphi_q$. Then from above: each $\varphi^k_q$ defines a diffeomorphism $\varphi^k_{q|\pi_{E}^{-1}(p)}:\pi_{E}^{-1}(p)\rightarrow \pi_{E}^{-1}(q)$. So let $U'\subset U$ be an open subneighbourhood of $B$. Then we can define a local trivialisation $\Psi:\pi_{E}^{-1}(p)\times U'\rightarrow \pi_{E}^{-1}(U')$, $(j^k_p\psi,q)\rightarrow \varphi^k_{q|\pi_{E}^{-1}(p)}(j^k_p\psi)$, which is smooth since the family of diffeomorphisms $\varphi_q$ is smoothly parametrised on $U'$ by construction, is bijective by construction, and has a smooth inverse (we can construct it using the inverses of the $\varphi_q$). Finally, by construction, $\pi_{J^kXM}\circ\Psi=\mathrm{proj}_{U'}$.
\end{proof}

\section*{Acknowledgments}
This material is based upon work supported by the National Science Foundation under Grant No. 2419967.  JOW is grateful to Bob Geroch, Hans Halvorson, David Malament, and JB Manchak for conversations on material related to this paper, and to an audience of the Southern California Philosophy of Physics group for a helpful discussion. EM acknowledges the support of Balliol College, Oxford, and the Faculty of Philosophy, University of Oxford. 

\bibliographystyle{elsarticle-harv}
\bibliography{covariance} 

\end{document}